%% file: normmax.tex
\documentclass[a4paper,english,11pt]{article}

\input{packageincludesE}

\begin{document}
\selectlanguage{english}
\title{Fixed Parameter Complexity and Approximability of \\ Norm Maximization\footnote{This work was initiated during the 10th INRIA-McGill workshop on Computational Geometry.}}
\author{Christian Knauer\footnote{Institut für Informatik, Universität Bayreuth, \texttt{christian.knauer@uni-bayreuth.de}} \and Stefan König\footnote{Zentrum Mathematik, Technische Universität München, \texttt{koenig@ma.tum.de}} \and Daniel Werner\footnote{Freie Universität Berlin
Fachbereich Mathematik und Informatik, \texttt{daniel.werner@fu-berlin.de}}}
\date{\today}

\maketitle

\begin{center}
\textbf{Abstract.}
\end{center}
\vspace{-0.9cm}
\paragraph{} The problem of maximizing the $p$-th power of a $p$-norm over a halfspace-presented polytope in $\R^d$ is a convex maximization problem which plays a fundamental role in computational convexity. It has been shown in \cite{ms-86} that this problem is $\NP$-hard for all values $p \in \mathbb{N}$, if the dimension $d$ of the ambient space is part of the input. In this paper, we use the theory of parametrized complexity to analyze how heavily the hardness of norm maximization relies on the parameter $d$.

More precisely, we show that for $p=1$ the problem is fixed parameter tractable 
but that for all $p \in \mathbb{N} \setminus \{1\}$ norm maximization is W[1]-hard.

Concerning approximation algorithms for norm maximization, we show that for fixed accuracy, there is a straightforward approximation algorithm for norm maximization in FPT running time, but there is no FPT approximation algorithm, the running time of which depends polynomially on the accuracy. 

As with the $\NP$-hardness of norm maximization, the W[1]-hardness immediately carries over to various radius computation tasks in Computational Convexity.

\section{Introduction and Preliminaries}

\paragraph{} The problem of computing geometric functionals of polytopes arises in many applications in mathematical programming, operations research, statistics, physics, chemistry or medicine (see e.g.~\cite{gritzmannKlee-handbook} for an overview). Hence, the question how efficiently these functionals can be computed or approximated has been studied extensively, e.g.~in \cite{bgkl-90, b-02, goodBadRadii, gk-93, ms-86}. 

Of particular interest is the problem of maximizing (the $p$-th power of) a $p$-norm over a polytope. Despite its simple formulation, this problem already exhibits the combinatorial properties which are responsible for hardness or tractability of the computation of many important geometric functionals. As for most computational problems on polytopes, the presentation of the input polytope is crucial for the computational complexity of norm maximization: If the input polytope is presented as the convex hull of finitely many points, norm maximization is solvable in polynomial time by the trivial algorithm of computing and comparing the norm of all these points. The situation changes dramatically when the input polytope is presented as the intersection of halfspaces. The present paper is concerned with the investigation of the parametrized complexity of this problem.

For $p\in \N\cup \{\infty\}$, a precise formulation of the norm maximization problem that we consider is as follows:

\begin{prb}[{\normmax[p]}] \label{prb:normmax}%
\begin{tabular}{ll}
\textbf{Input:} &$d \in \N$, $\gamma \in \Q$, rational $\CH$-presentation of a symmetric polytope $P \subseteq \R^d$ \\
\textbf{Parameter:} & $d$ \\
\textbf{Question:}& Is $\max \{\|x\|_p^p: x \in P \} \geq \gamma $?
\end{tabular}
\index{Normmax@\normmax[p]}
\end{prb}

Here, a rational $\CH$-presentation of a polytope is a presentation as intersection of finitely many halfspaces which are defined by inequalities that have only rational coefficients. 

\paragraph{} As shown in \cite{ms-86}, for $p= \infty$ (with the understanding that $\|x\|_\infty^\infty= \|x\|_\infty$), \normmax[\infty] is solvable in polynomial time via Linear Programming.  For all $p \in \N$, on the other hand, \normmax[p] is $\NP$-complete. (When speaking of $\NP$-hardness of parameterized problems, we mean the same decision problem, simply ignoring the parameter.) Moreover, in \cite{bgkl-90}, it is shown that $\NP$-hardness persists for all $p\in \N$ even when the instances are restricted to full-dimensional parallelotopes presented as a Minkowski sum of $d$ linearly independent line segments. Moreover, by \cite{b-02}, there is no polynomial time approximation algorithm for norm maximization for any constant performance ratio, unless $\P= \NP$.

\paragraph{} It is important to note that, as usual in the realm of computational convexity, the dimension $d$ is part of the input and the hardness of \normmax[p]  relies heavily on this fact, especially for the very restricted instances in \cite{bgkl-90}. Indeed, if $d$ is a constant, the obvious brute force algorithm of converting the presentation of $P$ yields a polynomial time algorithm with running time $O(n^d)$, where $n$ denotes the number of halfspaces in the presentation of $P$. However, this algorithm quickly becomes impractical as $n$ grows, even for moderate values of $d$.
The main purpose of this paper is to close the gap between $\NP$-hardness for unbounded dimension and a theoretically polynomial, yet impractical algorithm for fixed dimension.

A suitable tool that allows us to analyze how strongly the hardness of \normmax[p] depends on the parameter $d$ is the theory of Fixed Parameter Tractability. For an introduction to Fixed Parameter Tractability, we refer to the textbooks \cite{fg-06, n-06}. This theory has already been applied successfully to show the intractability of several problems in Computational Geometry even in low dimensions, see e.g.~\cite{pointSetPatternMatching, FPTkCenter, FPTzylinder,  FPTstabbing, FPTsurvey, christianFPTintro}.

\paragraph{} Our analysis of \normmax[p] shows that, although \normmax[p] is $\NP$-hard for all $p\in \N$, the hardness has a different flavor for different types of norms: Whereas hardness of \normmax[1] only comes with the growth of the dimension,  \normmax[p] has to be considered intractable already in small dimensions for all other values of $p$.

More precisely, we prove the following theorem:

\begin{theo}[Fixed-parameter complexity of \normmax] \label{theo:mainTheo}%
\normmax[1] is in FPT, whereas \normmax[p] is W[1]-hard for all $p\in \N \setminus \{1\}$.
\end{theo}

\paragraph{} The presented reduction also shows that in the hard cases no $n^{o(d)}$ algorithm for \normmax[p] exists, unless the \emph{Exponential Time Hypothesis}\footnote{The Exponential Time Hypothesis conjectures that n-variable 3-CNFSAT cannot be solved in $2^{o(n)}$-time; cf. \cite{ip-01}.} is false. Thus, the brute force algorithm for \normmax[p] mentioned above already has the best achievable complexity, if $p \in \N \setminus\{1\}$.

\paragraph{} In this case, one can also ask how strongly the inapproximability result of \cite{b-02} relies on the fact that \normmax[p] is a problem in unbounded dimension. For this purpose, call an algorithm that produces an $\bar{x}\in P$ such that, for some $\beta \in \N$,
$$ \|\bar{x}\|_p^p \geq \left(\frac{\beta-1}{\beta}\right)^p \max\{\|x\|_p^p : x \in P \}$$ 
a $\beta$-approximation-algorithm for \normmax[p]. The proof of the fact that \normmax[1] is in FPT then suggests the following: Replace the unit ball of the $p$-norm by a suitable symmetric polytope which approximates it sufficiently well and use the maximum of this polytopal norm as an approximation for the maximum of the $p$-norm. As polytopal norms can be maximized by solving a linear program for every facet of the unit ball and linear programs can be solved in $T_{LP}(d,n):= O(2^{2^d} n)$ (see \cite{MegiddoLPlinearFixedD}), which is polynomial in $n$ for fixed $d$, this yields an FPT-time approximation algorithm for fixed accuracy $\beta$.

\begin{theo}[Approximation complexity of \normmax] \label{theo:mainTheoApprox}%
Let $p\in \N\setminus\{1\}$. For every fixed $\beta \in \N$, there is a $\beta$-approximation-algorithm for \normmax[p] which 
runs in time $O(\beta^d T_{LP}(d,n))$. Conversely, there is no scheme of $\beta$-approximation-algorithms for \normmax[p] with running time $O(f(d)q(\beta, d, n))$ with a polynomial $q$ and an arbitrary computable function $f$.
\end{theo}

\paragraph{} Hence, although the problem is not in APX, approximation of \normmax[p] is possible for moderate values of $\beta$ and $d$. On the other hand, approximation tends to become costly as soon as the dimension or the desired accuracy grows.

\paragraph{} Finally, analogously to the $\NP$-hardness of \textsc{Normmax}$_p$, the W[1]-hardness of \textsc{Normmax}$_p$ implies the intractability of various problems in Computational Convexity as immediate corollaries. In Section~\ref{sec:corollaries}, we show that for the respective values of $p$, the problems \textsc{Circumradius}$_p$-$\CH$, \textsc{Diameter}$_p$-$\CH$, \textsc{Inradius}$_p$-$\CV$ and \textsc{Width}$_p$-$\CV$ (all parameterized by the dimension) are W[1]-hard.

\paragraph{} This paper is organized as follows. In the remainder of this section, we explain our notation. In Section~\ref{sec:complexityNormmax}, we will analyze the parameterized complexity of \normmax[p], i.e.~we prove Theorem~\ref{theo:mainTheo} and prepare some technical lemmas, which we will also use in Section~\ref{sec:approx} where we prove Theorem~\ref{theo:mainTheoApprox}. Finally, in Section~\ref{sec:corollaries}, we prove the corollaries for the mentioned radius computation tasks.

\paragraph{Notation.\\} The symbols $\N, \Z, \Q$ and $\R$\index{$\N$}\index{$\Z$}\index{$\Q$}\index{$\R$} are used to denote the set of positive integers, integers, rational numbers and real numbers, respectively.

For a positive integer $n\in \N$, we will abbreviate $[n]:= \{1,\dots, n\}$\index{$[n]$}. 

Throughout this paper, we are working in $d$-dimensional real space $\R^d$\index{$\R^d$} and for $A \subseteq \R^d$ we write $\lin(A)$\index{linear hull}\index{$\lin(\cdot)$}, $\aff(A)$\index{affine hull}\index{$\aff(\cdot)$}, $\conv(A)$\index{convex hull}\index{$\conv(\cdot)$}, $\pos(A)$\index{positive hull}\index{$\pos(\cdot)$}, $\int(A)$\index{interior}\index{$\int(\cdot)$}, $\relint(A)$\index{relative interior}\index{$\relint(\cdot)$}, and $\bd(A)$\index{boundary}\index{$\bd(\cdot)$} for the linear, affine, convex or positive hull and the interior, relative interior and the boundary of $A$, respectively. 

For a set $A \subseteq \R^d$, its dimension is $\dim(A):= \dim (\aff(A))$\index{$\dim(\cdot)$}\index{dimension}. Furthermore, for any two sets $A, B \subset \R^d$ and $\rho \in \R$, let $\rho A := \{\rho a: a \in A\}$ and $A+B:= \{a+b: a\in A, b\in B\}$ the $\rho$-\emph{dilatation}\index{dilatation} of $A$ and the \emph{Minkowski sum}\index{$A+B$ (for sets)}\index{Minkowski sum} of $A$ and $B$, respectively. 
We abbreviate  $A + (-B)$ by $A -B$ and  $A+\{c\}$ by $A+c$. 
A set  $K \subseteq \R^d$ is called 0-\emph{symmetric}\index{0-symmetric} if $-K= K$. If there is a $c \in \R^d$ such that $-(c+K)= c+K$ we call $K$ \emph{symmetric}\index{symmetric}.

If a polytope $P \in \CP^d$ is described as a bounded intersection of halfspaces, we say that $P$ is in $\CH$-presentation\index{H-presentation@$\CH$-presentation}. If $P$ is given as the convex hull of finitely many points, we call this a $\CV$-presentation\index{V-presentation@$\CV$-presentation} of $P$.
For a convex set $C\subseteq \R^d$, we let $\ext(C)$ denote  the set of  \emph{extreme points} of $K$.

For $1\leq p < \infty$, the $p$-norm\index{p-norm@$p$-norm} of a point $x= (x_1,\dots, x_d)^T \in \R^d$ is defined as \index{\nnn}

$$\|x\|_p := \left(\sum_{i=1}^d |x_i|^p\right)^{\frac{1}{p}}$$ 
for $p=\infty$, we let $\|x\|_\infty:= \max \{|x_i|: i \in [d]\}$.

For $p \in [1,\infty]$, we write $\B_p^d:= \{x \in \R^d: \|x\|_p \leq 1\}$\index{$\B_p^d$} for the unit ball\index{unit ball} of $\|\cdot\|_p$ and $\S_p^{d-1}:=\{x \in \R^d: \|x\|_p=1\}$\index{sphere}\index{$\S_p^{d-1}$} for the unit sphere in $\R^d$.

For two vectors $x,y\in \R^d$, we use the notation $x^Ty:= \sum_{i=1}^d x_i y_i$\index{$x^Ty$}\index{scalar product}\index{dot product}\index{inner product} for the standard scalar/inner/dot product of $x$ and $y$ and by 
$$H_{\leq}(a,\beta):= \{x \in \R^d: a^T x \leq \beta\}$$
we denote the half-space\index{halfspace}\index{$H_\leq(a,\beta)$} induced by $a\in \R^d$ and $\beta \in \R$, bounded by the hyperplane\index{hyperplane}\index{$H_=(a,\beta)$} $H_{=}(a,\beta):= \{x \in \R^d: a^T x = \beta\}$.

If $X$ is a finite set and $k\in \N$, then $\binom{X}{k}:= \{Y \subseteq X: |Y| = k\}$\index{$\binom{X}{k}$} denotes the set of all subsets of $X$ of cardinality $k$.

The standard basis in  $\R^d$ is denoted by  $\{e_i: i\in [d]\}$\index{$e_i$}\index{$\{e_1,\dots, e_d\}$}\index{standard basis}; the all-ones vector by $\ones:= (1, \dots, 1)^T \in \R^d$\index{$\ones$}.

\paragraph{} We denote by $\P$\index{$\P$} (and $\NP$\index{$\NP$}, respectively) the classes of decision problems that are solvable (verifiable, respectively) in polynomial time. For an account on complexity theory, we refer to \cite{GareyJohnson}. We write FPT\index{FPT}\index{$\text{FPT}$} for the class of fixed-parameter-tractable problems and W[1]\index{$W[1]$}\index{W[1]} for the problems of the first level of the W-hierarchy in the theory of Fixed Parameter Tractability. For an introduction to Fixed Parameter Tractability, we refer to the textbooks \cite{fg-06, n-06}.

\section{Fixed Parameter Complexity of Norm Maximization}
\label{sec:complexityNormmax}
\subsection{Intractability}
\label{sec:hardnessNormmax}
\paragraph{} We will first prove the hardness result for \normmax[p] for $p \geq 2$ via an FPT reduction of the W[1]-complete problem \textsc{Clique} to \normmax[p]. The formal parametrized decision problem of \textsc{Clique} is given in Problem~\ref{prb:clique}; a proof of its W[1]-completeness can be found e.g.~in \cite[Theorem 6.1]{fg-06}. 

\begin{prb}[\textsc{Clique}] \label{prb:clique}%
\begin{tabular}{ll}
\textbf{Input:} &$n,k \in \N$, $E \subseteq \left({[n]} \atop 2 \right)$ \\
\textbf{Parameter:} & $k$ \\
\textbf{Question:}& Does $G=([n], E)$ contain a clique of size $k$?
\end{tabular}
\end{prb}

Moreover, it is shown in \cite{Chen2005216} that \textsc{Clique}  cannot be solved in time $n^{o(k)}$, unless the Exponential Time Hypothesis fails.

\paragraph{} In order to show the hardness result, we will first show how to construct a polytope $P$ for a graph $G=([n], E)$ with the property that 
$$\max\{\|x\|_p^p: x\in P\} = k \quad \Longleftrightarrow \quad G \text{ contains a clique of size } k.$$ 
This \enquote{reduction} will be laid out as if irrational numbers were computable with infinite precision. The second part of this section will then show that the numbers can be rounded to a sufficiently rough grid in order to make the reduction suitable for the Turing machine model.

\paragraph{The construction.\\} 
Let $(n, k, E)$ be an instance of \textsc{Clique} and $p \in [1, \infty)$. Throughout this paper, we assume without loss of generality that $n$ is an even number. (If not, we add an isolated vertex to the graph.)

We choose $d:=2k$ and consider 
$$\R^{2k}= \R^2 \times \R^2 \times \ldots \times \R^2$$

i.e.~we will think of a vector $x \in \R^{2k}$ as $k$ two-dimensional vectors stacked upon each other. Therefore, it will be convenient to use the following notation.

\begin{nota}\label{nota:halfspaceR2}%
By indexing a vector $x \in \R^{2k}$, we refer to the $k$ two-dimensional vectors $x_1,\dots, x_k \in \R^2$ such that $x= (x_1^T, \dots, x_k^T)^T$. 
Further, for $a \in \R^2$ and $\beta \in \R$, we let\index{$H_\leq^i(a, \beta)$}
$$H_\leq^i(a, \beta):= \{x \in \R^{2k}: a^T x_i \leq \beta\}.$$
\end{nota}

\paragraph{} In order to construct an $\CH$-presentation of a polytope $P \subseteq \B_p^2 \times \B_p^2 \times \dots \times \B_p^2$, we will first construct a 2-dimensional polytope $P_1 \subseteq\B_p^2$ as our basic building block by placing vertices on the unit sphere $\S_p^1$ (compare Figure~\ref{fig:P1}):

For $v\in [\frac{n}{2}]$, let 
\begin{equation}
p_v':= \left(1\atop 0\right) + \frac{2(v-1)}{n} \left({-1} \atop{1}\right)\quad \text{ and }\quad \{p_v\}:= \left(p_v' + [0,\infty)\left(1 \atop 1\right)\right) \cap \S_p^1;
\label{eq:defiPv}
\end{equation}
for $v  \in [n]\setminus [\frac{n}{2}]$ let 
\begin{equation}
p_v':= \left(0\atop 1\right) + \frac{2v-(n+2)}{n} \left({-1} \atop {-1}\right)\quad \text{ and }\quad \{p_v\}:= \left(p_v' + [0,\infty)\left(-1 \atop 1\right)\right)\cap \S_p^1.
\label{eq:defiPv2}
\end{equation}

For $v\in [2n]\setminus [n]$, let 
$$p_v:=- p_{v-n}$$ 

and
\begin{equation}
P_1:= \conv\{p_1, \dots, p_{2n}\} = \bigcap_{v\in [2n]} H_\leq(a_v, \beta_v) \subseteq \R^2.
\label{eq:P1}
\end{equation}
Note that $P_1$ is 0-symmetric by construction and that the required $\CH$-presentation of $P_1$ in \eqref{eq:P1} can be computed in time $O(n\log(n))$, see e.g.~\cite{CGdeBerg}. For notational convenience, we also define 
$$p_{2n+1}:= p_1 \quad \text{ and } \quad  p_{-1}:= p_{2n}.$$
\begin{figure}[htb]
\centering
\includegraphics[width=0.4\textwidth]{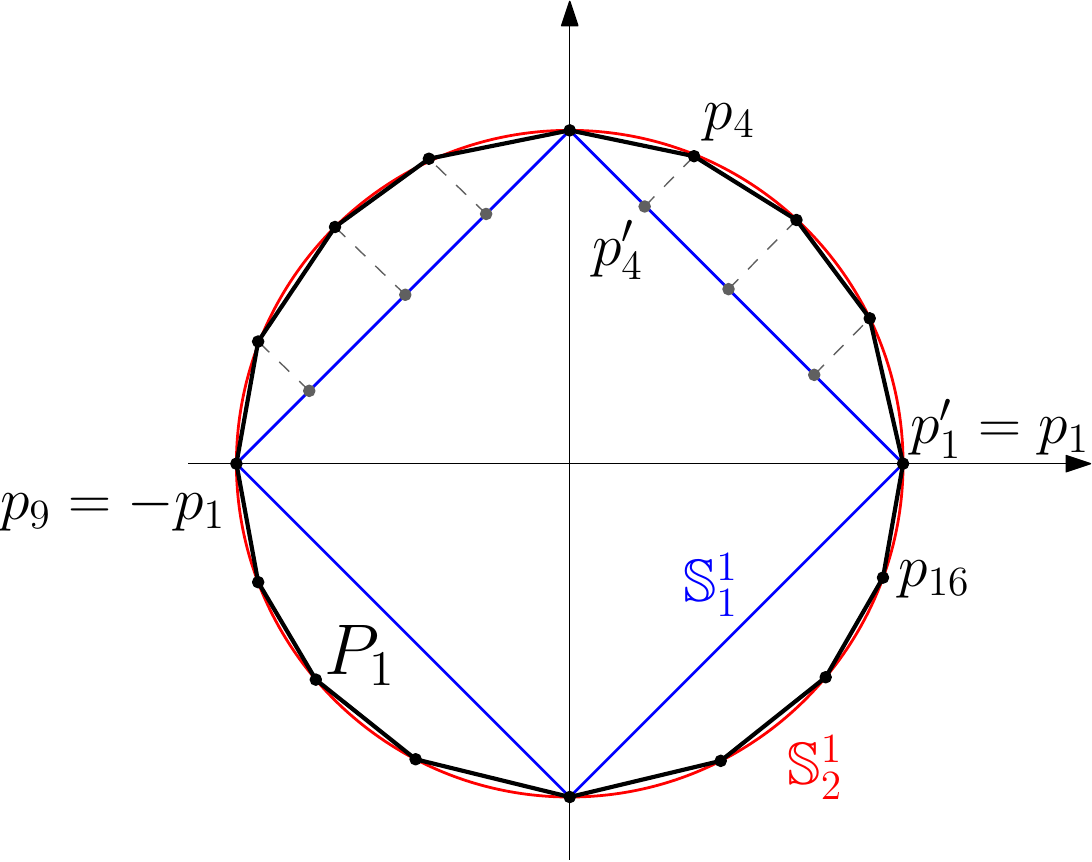}
\caption[Construction of $P_1$ (in the FPT reduction of \textsc{Clique}  to {\normmax})]{Construction of $P_1$ in the case $p= 2$, $n=8$.\label{fig:P1}}
\end{figure}

\begin{lem}[Distance between the $p_v$]
Let $P_1:= \conv\{p_1, \dots, p_{2n}\}$ be the polytope defined in Equation~\eqref{eq:P1} and $v\in[2n]$. The distance between two neighboring points on $\S_p^1$ satisfies
$$ \|p_v - p_{v+1}\|_2 \in \left[ \frac{2\sqrt{2}}{n}, \frac{4}{n}\right].$$
\label{lem:distances}
\end{lem}
\begin{proof}
Let $\Pi: \R^2\rightarrow \B_1^2$ denote the projection onto $\B_1^2$. By the definitions in \eqref{eq:defiPv} and \eqref{eq:defiPv2}, we have $\Pi(p_v)= p_v'$. Since $\Pi$ is contracting, the equidistant placement of $p_1',\dots, p_n'$ yields $\|p_v - p_{v+1}\|_2\geq \|p_v'-p_{v+1}'\|_2 = \frac{2\sqrt{2}}{n}$ for all $v\in [2n]$.

For the other bound, assume that $v \leq \frac{n}{4}$. (The other cases can be handled with the same arguments.) By elementary properties of $\B_p^2$, we have $e_1^T p_{v+1}\leq e_1^T p_v$ and $\ones^T p_{v+1} \geq \ones^T p_v$ and thus $p_{v+1} \in [q_1, q_2]$ with $q_1, q_2$ defined as in Figure~\ref{fig:distances}.

\begin{figure}[h]
\centering
\includegraphics[width= 0.37\textwidth]{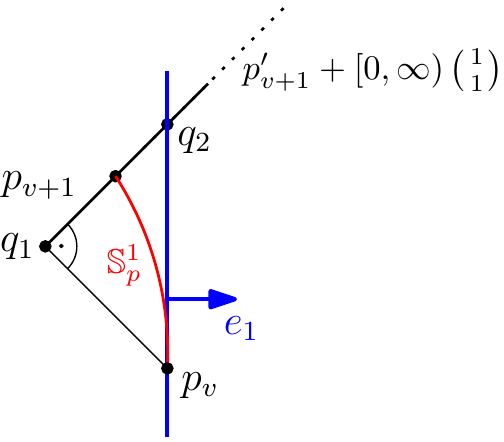}
\caption[Distances of points on unit sphere]{The situation in the proof of Lemma~\ref{lem:distances}.\label{fig:distances}}
\end{figure}

Inspection of the triangle $\conv \{p_v, q_1, q_2\}$ shows that it is equilateral with a right angle at $q_1$. Thus, $\|p_v -p_{v+1}\|_2\leq \|p_v -q_2\|_2 = \sqrt{2\|p_v-q_1\|_2^2}=\frac{4}{n}$.
\end{proof}

\paragraph{} Using Notation~\ref{nota:halfspaceR2}, we define a polytope $P_2 \subseteq \R^{2k}$ via
$$P_2:= \bigcap_{i\in [k]} \bigcap_{v \in [2n]} H_\leq^i(a_v, \beta_v) \subseteq \R^{2k}.$$

Observe that $P_2$ is 0-symmetric by construction and that any vertex $x$ of $P_2$ is of the form $x = (p_{v_1}, \dots, p_{v_k})^T$ for suitable $v_1,\dots, v_k \in [2n]$.

As for any $x= (x_1^T, \dots, x_k^T)^T \in \R^{2k}$ the identity 
$$\|x\|_p^p = \sum_{i=1}^k \|x_i\|_p^p $$
holds, and as for $p \in \N \setminus\{1\}$ the unit sphere $\{x \in \R^2: \|x\|_p^p=1\}$ contains no straight line segments, it follows that for $x \in P_2$,
$$ \|x\|_p^p \geq k \quad \Longleftrightarrow \quad x=\vector{p_{v_1}\\ \vdots\\ p_{v_k}} \text{ for some } v_1,\dots, v_k \in [2n] .$$

\paragraph{} For $v \in [2n]$, let $x_v, y_v \in \R$ be the coordinates of $p_v=(x_v,y_v)^T$ and define 
\begin{equation}
q_v:= \left({\mathrm{sgn}(x_v)|x_v|^{p-1}}\atop{\mathrm{sgn}(y_v)|y_v|^{p-1}}\right).
\label{eq:qU}
\end{equation}
 Noting that for all $x \in P_1$ and $v \in [2n]$, $q_v^T x = 1$ if and only if $x= p_v$, we define
\begin{equation}
\varepsilon:= 1 - \max \{ q_u^T p_v : u,v \in [2n], u\neq v\} > 0
\label{eq:epsilon}
\end{equation}
and for $u,v \in [n] $ and $i,j \in [k]$,
$$ E^{ij}_{uv}:= \{x \in\R^{2k}:\varepsilon - 2 \leq q_u^Tx_i + q_v^T x_j \leq 2- \varepsilon \} $$
and
$$ F^{ij}_{uv}:= \{x \in\R^{2k}:\varepsilon - 2 \leq q_u^Tx_i - q_v^T x_j \leq 2- \varepsilon \}.$$

Thus, if $x$ is a vertex of $P_2$ with $x_i = \pm p_u$ and $x_j= \pm p_v$ for some $u,v \in [n]$, then $x \not \in E^{ij}_{uv} \cap F^{ij}_{uv}$, i.e.~if $u,v\in [n]$ and $\{u,v\}\not \in E$ the constraints of $E^{ij}_{uv} \cap F^{ij}_{uv}$ make sure that $P$ does not contain a vertex with $x_i= \pm p_u$ and $x_j = \pm p_v$.

\paragraph{} Finally, to encode the \textsc{Clique} instance, we let $N:= \binom{[n]}{2}\setminus E$, define

$$P:= P_2 \cap \bigcap_{{{\{u,v\} \in N}} \atop {{i,j\in [k], i\neq j}}} (E^{ij}_{uv} \cap  F^{ij}_{uv})\cap \bigcap_{{v \in [n]}\atop {{i,j\in [k], i\neq j}}} (E^{ij}_{vv}\cap F^{ij}_{vv}),$$ 

and obtain the following lemma.

\begin{lem}[Reduction with infinite precision]
Let $(n, k, E)$ be an instance of \textsc{Clique}, $p \in [1, \infty)$ and $P \subseteq \R^{2k}$ the polytope obtained by the construction above. Then,
 $$\max\{\|x\|_p^p: x\in P\} = k \quad \Longleftrightarrow \quad G=([n], E) \text{ contains a clique of size } k.$$ 
\label{lem:infinitePrecision} 
\end{lem}

\paragraph{Analysis of the constructed polytope.\\} 
We will now investigate how much we can perturb the (possibly irrational) polytope $P$ in order to make it suitable for an FPT-reduction without loosing its ability to decide between Yes- and No-instances of \textsc{Clique}. 
For this purpose, we define the constant
\begin{equation}
 U:= \frac{1}{n^{2p}k^2}. 
\end{equation}

In the following, we show that rounding the vertices $p_1,\dots, p_{2n}$ of our initial polytope $P_1\subseteq \R^2$ to the grid $\frac{U}{2} \Z^2$ preserves all important features of our reduction. 
Since the parameter $p$ is a constant in \normmax[p], all the necessary computations can be carried out with a precision of $O(\log(nk))$ bits. Since we only need a polynomial number of computations, the whole reduction can be carried out in polynomial time.

\begin{lem}
Let $P_1= \conv\{p_1,\dots, p_{2n}\}\subseteq \R^2$ with $p_1,\dots, p_{2n} \in \S_p^1$ be the polytope from Equation \eqref{eq:P1}. For $\varepsilon:= 1 - \max \{ q_u^T p_v : u,v \in [2n], u\neq v\}$ with $q_u$ defined as in Equation \eqref{eq:qU}, we have
$$\varepsilon \geq \frac{2^{p-1}}{pn^p}.$$
\label{lem:epsilon}
\end{lem}
\begin{proof}
Let $x:=(x_1,x_2)^T\in \S^1_p$  and $y:=(y_1,y_2)^T \in \S_1^1$ with $x,y \geq 0$, $\|x-e_1\|_2 \geq \frac{2\sqrt{2}}{n}$, and $\|y-e_1\|_2 \geq \frac{2\sqrt{2}}{n}$. Since $\B_1^2 \subseteq \B_p^2$, $x_2 \geq y_2 \geq \frac2n$. Combining this inequality with $x \in \S_p^1$ yields
\begin{equation}
x_1 = (1-x_2^p)^{\frac{1}{p}}\leq \left(1-\left(\frac{2}{n}\right)^p\right)^{\frac{1}{p}} \leq 1- \frac{2^p}{pn^p},
\label{eq:x1}
\end{equation}
where the last inequality follows by bounding the concave function $x \mapsto x^{\frac{1}{p}}$ from above by a linear approximation at $x=1$.

Now, let $u,v \in [2n]$ with $u\neq v$. Then,
\begin{equation}
q_u^T p_v = q_u^T p_u + q_u^T (p_v - p_u) = 1 + \cos(q_u, p_v- p_u) \, \|q_u\|_2\,  \|p_v - p_u\|_2.
\label{eq:part1}
\end{equation} 

Since the points of lowest curvature on $\S_p^1$ are $\pm e_1$ and $\pm e_2$, and since $e_1= p_1 = q_1$, we obtain $\cos(q_u, p_v- p_u) \leq \cos(e_1, p_2 - e_1)$, which in turn can be bounded by 
$$\cos(e_1, p_2 - e_1) \leq \frac{x_1-1}{\|p_2-e_1\|_2}$$ with $x_1= e_1^T x$ for the point $x\in \S_1^p$ defined above.  
Further, $q_u \in \S_{\frac{p}{p-1}}^1$ implies $\|q_u\|_2 \geq \frac{\sqrt{2}}{2}$, and $ \|p_v - p_u\|_2 \geq \frac{2\sqrt{2}}{n}$ by Lemma~\ref{lem:distances}.
Using \eqref{eq:x1}, we can continue Equation \eqref{eq:part1} to
$$q_u^T p_v \leq 1 - \frac{2^p}{pn^p \|p_2 - e_1\|}\cdot \frac{\sqrt{2}}{2}\cdot \frac{2\sqrt{2}}{n} \leq 1 - \frac{2^{p-1}}{pn^p}, $$
where the last inequality follows again from Lemma~\ref{lem:distances}.
\end{proof}

\paragraph{} For $v \in [n]$, let $\bar p_v$ be the rounding of $p_v$ to the grid $\frac{U}{2} \Z^2$ and define $\bar p_v = -\bar p_{v-n}$ for $v\in [2n]\setminus[n]$ and further 
\begin{equation}
\bar P_1:= \conv\{\bar p_1,\dots, \bar p_{2n}\}.
\label{eq:P1bar}
\end{equation} For $\bar p_v = (\bar x_v, \bar y_v)^T \in \R^2$, define 
$$\bar q_v:= \left({\mathrm{sgn}(\bar x_v) \, |\bar x_v|^{p-1}}\atop {\mathrm{sgn}(\bar{y}_v)\,|\bar y_v|^{p-1}}\right).$$

By choice of our grid, we get 
\begin{equation}
\|p_v -\bar p_v\|_{p'} \leq  U \quad \forall p' \geq 1.
\label{eq:normDifference}
\end{equation} Moreover, if $q \in [1, \infty)$ is such that $\frac{1}{p}+\frac{1}{q} =1$, then $\|q_v\|_q =1$ for all $v\in [2n]$ and since $x \mapsto x^{p-1}$ is Lipschitz continuous on $[-1,1]$ with Lipschitz constant $L= p-1$, we obtain 
\begin{equation}
\|q_v - \bar q_v \|_1 \leq (p-1) U.
\label{eq:normalsDifference}
\end{equation}

\paragraph{} First, we show that the points $\bar p_1,\dots, \bar p_{2n}$ are still in convex position, which is binned into a separate lemma.
\begin{lem}
Let $\bar P_1 = \conv\{\bar p_1,\dots, \bar p_{2n}\}\subseteq \R^2$ the polytope from \eqref{eq:P1bar}. Then, $\ext(\bar P_1)= \{\bar p_1, \dots, \bar p_{2n}\}$ and the coding length of an $\CH$-presentation of $\bar P_1$ is polynomially bounded in the coding length of $\bar p_1,\dots, \bar p_{2n}$.
\label{lem:barP1}
\end{lem}
\begin{proof}
 For $v \in [2n]$, we have $q_v^T \bar p_{v} \geq 1- \|q_v\|_2 U \geq 1- \|q_v\|_q U = 1-U$, since $p\geq 2$ and therefore $q \leq 2$. For $u \in [2n]\setminus \{v\}$, we get $q_v^T \bar p_{u} \leq 1-\varepsilon + U$. Since $1-\varepsilon + U < 1-U$, the hyperplane $H_=(q_v, 1- \varepsilon +U)$ separates $\bar p_v$ from $\conv (\{\bar p_1,\dots, \bar p_{2n}\}\setminus \{\bar p_v\})$ and hence $\bar p_v \in \ext(\bar P_1)$.

 Assume now that $\bar P_1:= \{x \in \R^2: \bar a_v^T x \leq 1 ~\forall v \in [2n]\}$ is an $\CH$-presentation of $\bar P_1$. Applying Cramer's Rule, we see that, for all $v \in [2n]$, the entries of $\bar a_v$ are quotients of polynomials in $\bar p_1,\dots, \bar p_{2n}$ and so the coding length of the $\CH$-presentation of $\bar P_1$ is bounded by a polynomial in the coding length of $\bar p_1,\dots, \bar p_{2n}$.
\end{proof}

\paragraph{} Since the coding length of $\bar P_1$ is polynomially bounded, we also get that the coding length of 
$$\bar P_2:= \bigcap_{i\in [k]} \bigcap_{v\in [2n]} H_\leq^i (\bar a_v, \bar \beta_v) \subseteq \R^{2k}. $$
is polynomially bounded.

\paragraph{} Now, let $\bar \varepsilon := 1 - \max \{\bar q_u^T \bar p_v: u,v \in [2n], u\neq v\}$. By expanding the expression $\bar q_u^T \bar p_v =\bigl(q_u + (\bar q_u -q_u)\bigr)^T \bigl(p_v +(\bar p_v - p_v)\bigr)$ and using \eqref{eq:normDifference} and \eqref{eq:normalsDifference}, we obtain 
\begin{equation}
\bar \varepsilon \geq \varepsilon - 3pU> 0.
\label{eq:barvarepsilon}
\end{equation}

Finally, define $$ \bar E^{ij}_{uv}:= \{x \in\R^{2k}:\bar \varepsilon - 2 \leq \bar p_u^Tx_i + \bar p_v^T x_j \leq 2- \bar \varepsilon \}, $$
and $$ \bar F^{ij}_{uv}:= \{x \in\R^{2k}:\bar \varepsilon - 2 \leq \bar p_u^Tx_i -, \bar p_v^T x_j \leq 2- \bar \varepsilon \}, $$
and, for $N:= \binom{[n]}{2}\setminus E$, let

\begin{equation}
\bar P:= \bar P_2 \cap \bigcap_{{{\{u,v\} \in N }} \atop {{i,j\in [k], i\neq j}}} (\bar E^{ij}_{uv} \cap  \bar F^{ij}_{uv})\cap \bigcap_{{v \in [n]}\atop {{i,j\in [k], i\neq j }}} (\bar E^{ij}_{vv}\cap \bar F^{ij}_{vv}).
\label{eq:barP}
\end{equation}

\paragraph{} The following two lemmas will now prepare the proof that we can still reduce \textsc{Clique} to norm maximization over $\bar P$. To be able to state them in a concise way, we introduce the following notation.

\begin{nota}
Let $\bar P\subseteq \R^{2k}$ be the polytope from Equation \eqref{eq:barP} and $x = (x_1^T, \dots, x_k^T)^T\in \bar P$. 
By letting\index{$m_i(x)$}
$$m_i(x) \in \argmax \{\bar q_v^T x_i: v \in [2n]\}, $$
we can refer to the index of a vertex which is \enquote{closest} to $x$ in the sense that $\bar q_{m_i(x)}^T x \geq \bar q_v^T x$ for all $v \in [2n]$.  This is illustrated in Figure~\ref{fig:mI}.
\begin{figure}[htb]
\centering
\includegraphics[width=0.7\textwidth]{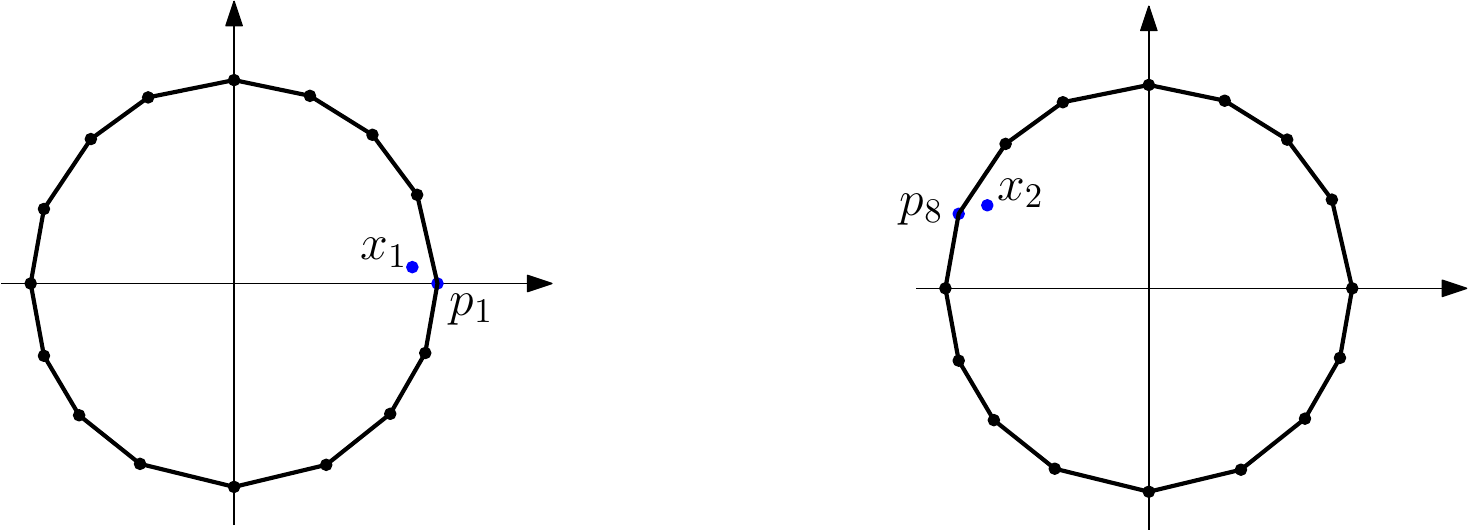}
\caption[Illustration of Notation~\ref{nota:mI}]{Illustration of Notation~\ref{nota:mI}. The figure shows a point $x= (x_1^T, x_2^T)^T\in \R^4$ with $m_1(x)= 1$ and $m_2(x)=8$.\label{fig:mI}}
\end{figure}
\label{nota:mI}
\end{nota}

\paragraph{} First, we show that if $\bar P$ contains a point which is \enquote{close} (in the sense specified in Notation~\ref{nota:mI}) to a clique vertex, then $\bar P$ contains the clique vertex itself.
\begin{lem}
Let $\bar P\subseteq \R^{2k}$ be the polytope constructed above in Equation \eqref{eq:barP}. If there exists $\bar x \in \bar P$ such that 
$ \bar q_{m_i(\bar x)}^T \bar x > 1 -\frac{\bar \varepsilon}{2}$ for all $i \in [k],$ then $(\bar p_{m_1(\bar x)}^T, \dots, \bar p_{m_k(\bar x)}^T)^T \in \bar P$.
\label{lem:notEverywhereClose}
\end{lem}
\begin{proof}
Since $\bar q_{m_i(\bar x)}^T \bar x_i > 1-\frac{\bar \varepsilon}{2}$ for all $i \in [k]$, for no pair $(i,j) \in [k]^2$ the inequalities $\bar q_{m_i(\bar x)}^T x_i + \bar q_{m_j(\bar x)}^T x_j \leq 2-\bar \varepsilon$ can be present in the description of $\bar P$. Since, by definition of $\bar \varepsilon$, we have $\bar q_v^T \bar p_{m_i(\bar x)} \leq 1-\bar \varepsilon$ for all $v \in [2n]\setminus\{m_i(\bar x)\}$ and $i\in [k]$, we can conclude that
$$(\bar p_{m_1(\bar x)}^T,\dots, \bar p_{m_k(\bar x)}^T)^T \in \bar P.$$ 
\end{proof}

\paragraph{} In view of Lemma~\ref{lem:notEverywhereClose}, it remains to show that the norm of a vertex which is \enquote{far} from a clique vertex is sufficiently small: 

\begin{lem}
Let $v \in [2n]$ and $Q:= \conv\{0, \bar p_v, \bar p_{v+1}\} \cap H_\leq (\bar q_v, 1-\frac{\bar \varepsilon}{2}) \cap H_\leq (\bar q_{v+1}, 1-\frac{\bar \varepsilon}{2})$. Then, for $n$ sufficiently large, 
$$\max \{\|x\|_p^p: x \in Q \} \leq 1 - \frac{2^{p-3}}{pn^p}.$$
\label{lem:far}
\end{lem}
\begin{proof}
Let $Q':=\conv\left\{0, e_1, \bar p_2 \right\}\cap H_\leq \left(e_1, 1-\frac{\bar \varepsilon}{2}\right) \cap H_\leq \left( \bar q_2, 1-\frac{\bar \varepsilon}{2}\right)$. Since $e_1$ is a point of lowest curvature on the boundary of $\B_p^2$, we have  $ \max \{\|x\|_p^p: x \in Q \} \leq \max \{\|x\|_p^p: x \in  Q'\}= \|x^*\|_p^p$, where $x^*$ fulfills $e_1^T x^*= 1-\frac{\bar \varepsilon}{2}$ and $x^*= \lambda e_1 + (1-\lambda)\bar p_2$ for some $\lambda \in [0,1]$. From the first property, we can deduce $\lambda = 1-\frac{\bar \varepsilon}{2}$, which implies $e_2^T x^*= \frac{\bar \varepsilon}{2} e_2^T \bar p_2$. By Lemma~\ref{lem:distances}, $e_2^T \bar p_2 \leq \frac{2}{n}+U$. Putting things together, we obtain
\begin{equation}
 \|x^*\|  \leq \left(1-\frac{\bar \varepsilon}{2}\right)^p + \left(\frac{\bar \varepsilon}{2}\left(\frac{2}{n}+U\right)\right)^p \leq \left(1-\frac{\bar \varepsilon}{2}\right) +\left(\frac{\bar \varepsilon}{2}\left(\frac{2}{n}+U\right)\right)^p. 
 \label{eq:normFar}
\end{equation}
By Lemma~\ref{lem:epsilon} and Equation \eqref{eq:barvarepsilon}, $\bar \varepsilon \geq \frac{2^{p-1}}{pn^p} - 3pU$. By the choice of $U$ and the assumption that $n$ is sufficiently large, we can therefore continue \eqref{eq:normFar} and obtain

$$ \left(1-\frac{\bar \varepsilon}{2}\right)+\left(\frac{\bar \varepsilon}{2}\left(\frac{2}{n}+U\right)\right)^p  \leq  1-\frac{2^{p-3}}{pn^p}. $$

\end{proof}

\paragraph{Hardness part of Theorem~\ref{theo:mainTheo}.\\}
The following lemma shows that it is sufficient to carry out the reduction described by Lemma~\ref{lem:infinitePrecision} with finite precision as described in this subsection. It completes the proof of the hardness part of Theorem~\ref{theo:mainTheo}.  For notational convenience, we use the clique number\index{clique number} $\omega(G)$\index{$\omega(G)$} to denote the size of the biggest clique in a graph $G=([n], E)$.

\begin{lem}[Reduction with finite precision]
Let $(n,k,E)$ be an instance of \textsc{Clique}, $G=([n], E)$ and $\bar P \subseteq \R^{2k}$ the polytope with rounded coordinates constructed above in \eqref{eq:barP}. Then, 
\begin{equation}
\omega(G)\geq k \quad \Longleftrightarrow \quad   \max\{\|x\|_p^p: x\in \bar P\} \geq  k(1-U)^p
\label{eq:yes}
\end{equation}
and 
\begin{equation}
\omega(G)< k  \quad \Longleftrightarrow \quad \max\{\|x\|_p^p: x\in \bar P\} \leq  (k-1)(1+U)^p +1 -\frac{2^{p-3}}{pn^p}.
\label{eq:no}
\end{equation}
\label{lem:finitePrecision}
\end{lem}
\begin{proof}
Since $(k-1)(1+U)^p +1 -\frac{2^{p-3}}{pn^p} < k(1-U)^p$, if suffices to show the \enquote{forward} direction in both \eqref{eq:yes} and \eqref{eq:no}.

If $\omega(G) \geq k$ and $\{v_1, \dots, v_k\}\subseteq [n]$ is the vertex set of a $k$-clique in $G$, then $\bar P$ contains the vertex  $x^*= (\bar p_{v_1}^T,\dots, \bar p_{v_k}^T)^T$ and $\|x^*\|_p^p \geq k(1-U)^p$ by \eqref{eq:normDifference}.

Assume now that $\omega(G) < k$ and let $x^* \in \bar P$ be a vertex of maximal norm in $\bar P$. If $ \bar q_{m_i(x^*)}^T  x^* > 1 -\frac{\bar \varepsilon}{2}$ for all $i \in [k],$ Lemma~\ref{lem:notEverywhereClose} would imply that $(\bar p_{m_1(x^*)}^T, \dots, \bar p_{m_k(x^*)}^T)^T$ is a vertex of $\bar P$ and therefore contradict  $\omega(G) < k$. Hence, there is some $i \in [k]$ such that $\bar q_{m_i(x^*)}^T x^* \leq 1- \frac{\bar \varepsilon}{2}$. By adding a constant number of vertices to $G$, we can assume that $n$ is sufficiently large and apply Lemma~\ref{lem:far} in order to obtain  $\|x^*_i\|_p^p \leq 1 -\frac{2^{p-3}}{pn^p}$. As $\|x_j^*\|_p^p \leq (1+U)^p $ for all $j \in [k]\setminus\{i\}$, the right hand side of \eqref{eq:no} follows.
\end{proof}

\paragraph{} The construction of the polytope $P$ (or $\bar P$) relies on the fact that, for $p \geq 2$, the boundary of the unit ball of a $p$-norm contains no straight line segment. This is not the case for $p=1$ and we show in the next subsection that \normmax[1] is indeed in FPT.

\subsection{Tractability}
\paragraph{} This subsection completes the proof of Theorem~\ref{theo:mainTheo} by showing that \normmax[1] is fixed parameter tractable. 

The statement of Theorem~\ref{theo:normmaxFPT} is slightly more general than needed for Theorem~\ref{theo:mainTheo} but will be of use in Section~\ref{sec:approx}. The result for \normmax[1] can be obtained from Theorem~\ref{theo:normmaxFPT} by choosing $\varphi_d: \R^d \rightarrow \R; x \mapsto \|x\|_1$ in Problem~\ref{prb:maxPhi}.

\begin{prb}[\textsc{Max}-$\Phi$] \label{prb:maxPhi}%
Suppose that for each $d\in \N$, $\varphi_d: \R^d \rightarrow \R$ is positive homogeneous of degree 1 and let $\Phi:= (\varphi_d)_{d\in \N}$.
The problem \textsc{Max}-$\Phi$ is defined as follows:
\medskip

\begin{tabular}{ll}
\textbf{Input:} & $d \in \N$, $\gamma \in \Q$, rational $\CH$-presentation of a polytope $P\subseteq \R^d$ \\
\textbf{Parameter:} & $d$ \\
\textbf{Question:}& Is $\max \{\varphi_d(x) : x \in P\} \geq \gamma$?
\end{tabular}
\end{prb}

\begin{theo}[Tractability of \textsc{Max}-$\Phi$] \label{theo:normmaxFPT}%
For each $d\in \N$, let $\varphi_d: \R^d \rightarrow \R$ be positive homogeneous of degree 1 and $\Phi:= (\varphi_d)_{d\in \N}$. Suppose that, for $d \in \N$, the set $\B^d:= \{x\in\R^d: \varphi_d(x)\leq 1\}$ is a full-dimensional polytope, a rational $\CH$-presentation of which can be computed in time $f(d)$ for a computable function $f:\N \rightarrow \N$. Then, 
\textsc{Max}-$\Phi$ is in FPT and can be solved in time $O(f(d)T_{LP}(d,n))$.
\end{theo}
\begin{proof}
Let $\B^d= \bigcap_{i=1}^m H_\leq (a_i, 1)$ be an $\CH$-presentation of $\B^d$. Then, $m \in O(f(d))$. Because of the homogeneity of $\varphi_d$, 
$\{x \in \R^d: \varphi_d(x) \leq \lambda\} = \lambda \B^d$ and $\varphi_d(x)= \max_{i\in [m]} a_i^T x$.
Hence,
$$\max \{\varphi_d(x) : x \in P\} = \max_{i\in [m]} \max \{a_i^T x: x \in P\}.$$
Thus, \textsc{Max}-$\Phi$ can be decided by the following algorithm:

\begin{compactenum}[(1)]
\item Compute an $\CH$-presentation of $\B^d$ in time $f(d)$. 
\item Solve $m$ linear programs $\max\{a_i^T x : x \in P\}$ in time $T_{LP}(d,n)$. 
\item Compare the biggest objective value to $\gamma$.
\end{compactenum}
\vspace{0.1cm}

As $T_{LP}(d,n) \in O(2^{2^d}n)$, the above algorithm has FPT running time $O(f(d)2^{2^d}n)$.
\end{proof}

\paragraph{} We  can also establish fixed parameter tractability for the two problems $[-1,1]$-\textsc{Parmax}$_p$ and $[0,1]$-\textsc{Parmax}$_p$ as considered in \cite{bgkl-90}.

\begin{prb}[{$[0,1]$-\textsc{Parmax}$_p$}] \label{prb:parmax01}%
\begin{tabular}{ll}
\textbf{Input:} &$d \in \N$, $\gamma\in \Q$, $v_1, \dots, v_n \in \Q^d$ linearly independent \\
\textbf{Parameter:} & $d$ \\
\textbf{Question:}& Is $\max \{\|x\|_p^p: x \in \sum_{i=1}^d [0,1] v_i \} \geq \gamma $?
\end{tabular}
\end{prb}

\begin{prb}[{$[-1,1]$-\textsc{Parmax}$_p$}] \label{prb:parmax-11}%
\begin{tabular}{ll}
\textbf{Input:} &$d\in \N$, $\gamma \in \Q$, $v_1, \dots, v_n \in \Q^d$ linearly independent \\
\textbf{Parameter:} & $d$ \\
\textbf{Question:}& Is $\max \{\|x\|_p^p: x \in \sum_{i=1}^d [-1,1] v_i \} \geq \gamma $?
\end{tabular}
\end{prb}

\paragraph{} In \cite{bgkl-90}, it was shown that Problem~\ref{prb:parmax01} and \ref{prb:parmax-11} are both $\NP$-hard, so that the $\NP$-hardness of \normmax[p] persists even on very restricted instances. However, the following theorem shows that these problems are fixed parameter tractable, when parametrized by the dimension. So in this case, the hardness of \textsc{Parmax}$_p$ is really a phenomenon of high dimensions.

\begin{theo}[Tractability of \textsc{Parmax}$_p$]
For all $p \in \N$, Problems~\ref{prb:parmax01} and \ref{prb:parmax-11} are in FPT.
\end{theo}
\begin{proof}
We only consider Problem~\ref{prb:parmax01}; the argument for Problem~\ref{prb:parmax-11} is exactly the same. The vertices of the polytope $P:= \sum_{i=1}^d [0,1] v_i$ are all of the form $\sum_{i=1}^d \lambda_i v_i$ for some vector $\lambda= (\lambda_1, \dots, \lambda_d)^T \in \{0,1\}^d$. As the the maximum of $\|\cdot\|_p^p$ is attained at a vertex of $P$, it suffices to compute the norm of all $2^d$ possible choices of $\lambda \in \{0,1\}^d$. This clearly is an FPT-algorithm for Problem~\ref{prb:parmax01}.
\end{proof}

\section{Approximation}
\label{sec:approx}
\subsection{FPT-Approximation for Fixed Accuracy}
\paragraph{} In \cite{b-02}, it is shown that, for all $p \in \N$, {\normmax[p]} is not contained in APX (i.e.~there is no polynomial time approximation algorithm with a fixed performance guarantee). As norm maximization with a polytopal unit ball is in FPT, we can give a straightforward approximation algorithm that has FPT running time for any fixed accuracy by replacing the unit ball $\B_p^d$ by an approximating polytope. The following proposition concerning the complexity of such a polytope can be obtained from \cite[Lemmas 3.7 and 3.8]{br-01}.

\begin{prop}[Approximation of balls by polytopes]
Let $p\in \N$ and $\beta \in \N$ be fixed. There is a symmetric polytope $B \subseteq \R^d$ with a rational $\CH$-presentation and at most $O(\beta^d)$ facets such that 
\begin{equation} 
\B_p^d \subseteq B \subseteq \frac{\beta}{\beta-1} \B_p^d,
\label{eq:ballInclusion}
\end{equation}
and $B$ can be computed in time $O(\beta^d)$.
\label{prop:ballApprox}
\end{prop}

\begin{lem}[FPT-Approximation algorithm for fixed accuracy]
Let $p\in \N$ and $\beta \in \N$ be fixed. There is an algorithm which for every $\CH$-presented polytope $P \subseteq \R^d$ runs in time $O(\beta^d T_{LP}(d,n))$ and produces an $\bar{x}\in P$ such that
$$ \|\bar{x}\|_p^p \geq \left(\frac{\beta-1}{\beta}\right)^p \max\{\|x\|_p^p : x \in P \}.$$
\label{lem:approxFPT}
\end{lem}
\begin{proof}
The following algorithm has the desired properties:

\begin{enumerate}[(1)]
\item Compute an $\CH$-presentation of a symmetric polytope $B\subseteq \R^d$ with the properties of Proposition~\ref{prop:ballApprox} and let $\|\cdot\|_B: \R^d \rightarrow \R; x \mapsto \|x\|_B:= \min\{\lambda\geq 0: x \in \lambda B\}$ 
\item Choose $\bar{x} \in \argmax\{\|x\|_B: x \in P\}$.
\end{enumerate}

It follows from Proposition~\ref{prop:ballApprox} that step (1) can be accomplished in time $O(\beta^d)$. As the number of facets of $B$ is in $O(\beta^d)$, it follows from Theorem~\ref{theo:normmaxFPT} that the maximization of $\|\cdot\|_B$ over $P$ can be done in time $O(\beta^d T_{LP}(d,n))$.

In order to show the performance ratio of the above algorithm, observe that Property~\eqref{eq:ballInclusion} of $B$ implies that $\frac{\beta -1}{\beta} \|x\|_p \leq \|x\|_B\leq \|x\|_p$ for all $x \in \R^d.$
Hence, if $x^* \in \argmax\{\|x\|_p^p: x \in P\}$, we get
$$ \|\bar{x}\|_p^p \geq \|\bar{x}\|_B^p \geq\|x^*\|_B^p \geq \left(\frac{\beta -1}{\beta}\right)^p \|x^*\|_p^p= \left(\frac{\beta -1}{\beta}\right)^p \max\{\|x\|_p^p: x \in P\}.$$

\end{proof}

\subsection{No FPT-approximation for Variable Accuracy}
\paragraph{} Finally, we will show that the straightforward approximation of the previous subsection is already best possible in the sense that there is no algorithm with polynomial dependence on the approximation quality and exponential dependence only on the dimension. Hence, combined with Lemma~\ref{lem:approxFPT}, Lemma~\ref{lem:NoBetterApprox} completes the proof of Theorem~\ref{theo:mainTheoApprox}. 
In fact, the basis for this has already been established in Lemma~\ref{lem:finitePrecision} and we can give the result right away.

\begin{lem}[No polynomial dependence on $\beta$]
Let $f: \N \rightarrow \R$ be a computable function and $q: \R^3 \rightarrow \R$ a polynomial function. If W[1]$\neq$FPT, there is no algorithm which for every $\CH$-presented polytope $P \subseteq \R^d$ runs in time $O(f(d)q(\beta, d, n))$ and produces an $\bar{x}\in P$ such that
$$ \|\bar{x}\|_p^p \geq \left(\frac{\beta-1}{\beta}\right)^p \max\{\|x\|_p^p : x \in P \}.$$
\label{lem:NoBetterApprox}
\end{lem}
\begin{proof}
 Let $(n, k, E)$ be an instance of the W[1]-hard problem \textsc{Clique} and $\bar P\subseteq \R^{2k}$ the polytope constructed in Equation \eqref{eq:barP}. By Lemma~\ref{lem:finitePrecision}, it can be decided if $G=([n], E)$ has a clique of size $k$ by determining, whether 
\begin{equation}
\label{eq:simplifiedDecision}
\begin{array}{rl}
\textbf{either}& \max\{\|x\|_p^p : x \in \bar P\} \geq k(1-U)^p \\
\textbf{or}& \max\{\|x\|_p^p : x \in \bar P \} \leq (k-1)(1+U)^p +1 - \frac{2^{p-3}}{pn^p}
\end{array}
\end{equation}
Assume that an algorithm with the claimed properties exists and call it $\CA$. One easily checks that there is a suitable constant $C >0$ such that it suffices to choose $\beta \geq \frac{pn^pk}{C}$ in order to fulfill
$$\left(\frac{\beta}{\beta-1}\right)^p \left((k-1)(1+U)^p +1 - \frac{2^{p-3}}{pn^p}\right) < k(1-U)^p.$$ 

Hence, we can run the following algorithm $\CA'$ in order to decide \eqref{eq:simplifiedDecision}:
\begin{enumerate}[1)]
\item Choose $\beta:= \left \lceil \frac{pn^pk}{C} \right \rceil$. 
\item Run $\CA$ on the polytope $\bar P$ and obtain an approximate normmaximal vertex $\bar{x} \in \bar P$.
\item If $\|\bar{x}\|_p^p  > (k-1)(1+U)^p +1 - \frac{2^{p-3}}{pn^p}$, decide $\max\{\|x\|_p^p : x \in \bar P\} \geq k(1-U)^p$.
\item[] Else, decide $\max\{\|x\|_p^p : x \in P \} \leq (k-1)(1+U)^p +1 - \frac{2^{p-3}}{pn^p}$.
\end{enumerate}
By the properties of $\CA$, the running time of the algorithm $\CA'$ is $O(f(d)q(n^{p}k, d, n))$ and by  Lemma~\ref{lem:finitePrecision} and the choice of $\beta$, $\CA'$ decides \eqref{eq:simplifiedDecision} correctly. $\CA'$ is thus an FPT algorithm for \textsc{Clique}. Unless FPT=W[1], this is a contradiction to the fact that \textsc{Clique} is W[1]-hard.
\end{proof}

\section{Some Implications}
\label{sec:corollaries}

\paragraph{}As stated in the introduction, norm maximization over polytopes plays a fundamental role in Computational Convexity. This section gives corollaries concerning the hardness of determining four important geometric functionals on polytopes.

If $P\subseteq \R^d$ is a polytope, we denote by $R(P,\B_p^d)$ ($r(P, \B_p^d)$, respectively) the circumradius\index{circumradius} (inradius)\index{inradius} of $P$ with respect to the $p$-norm. Further, similar to the notation in \cite{gk-92}, we write $R_1(P,\B_p^d)$ ($r_1(P,\B_p^d)$) for half of the width\index{width} (diameter)\index{diameter} of $P$, i.e.~half the radius of a smallest slab containing $P$ (half the length of the longest line segment contained in $P$).

For $p\in \N\cup\{\infty\}$, we consider the following problems:

\begin{prb}[\textsc{Circumdadius}$_p$-$\CH$] \label{prb:circumRadius}%
\begin{tabular}{ll}
\textbf{Input:} &$d\in \N$, $\gamma \in \Q$, rational $\CH$-presentation of a 0-symmetric polytope $P \subseteq \R^d$ \\
\textbf{Parameter:} & $d$ \\
\textbf{Question:}& Is $R(P,\B_p^d)^p \geq \gamma $?
\end{tabular}
\index{Circumradius@\textsc{Circumdadius}$_p$-$\CH$}
\end{prb}

\begin{prb}[\textsc{Diameter}$_p$-$\CH$] \label{prb:diameter}%
\begin{tabular}{ll}
\textbf{Input:} &$d \in \N$, $\gamma \in \Q$, rational $\CH$-presentation of a 0-symmetric polytope $P \subseteq \R^d$ \\
\textbf{Parameter:} & $d$ \\
\textbf{Question:}& Is $r_1(P, \B_p^d)^p \geq \gamma $?
\end{tabular}
\index{Diameter@\textsc{Diameter}$_p$-$\CH$}
\end{prb}

\paragraph{} It has been shown in \cite{gk-93} that Problems~\ref{prb:circumRadius} and \ref{prb:diameter} are solvable in polynomial time if $p=\infty$ and, by using an identity for symmetric polytopes from \cite{gk-92}, that both problems are $\NP$-hard when $p \in \N$. Using the same identity, we can establish (in-)tractability for both problems when parameterized by  the dimension:

\begin{cor}[Circumradius {\&} Diameter]\label{cor:radiiHard}%
For $p=1$, Problems~\ref{prb:circumRadius} and \ref{prb:diameter} are in FPT. For $p\in \N\setminus \{1\}$, both problems are W[1]-hard.

\end{cor}
\begin{proof}
As shown in \cite[(1.3)]{gk-92}, for a 0-symmetric polytope $P\subseteq \R^d$, we have 
$$ R(P,\B_p^d)^p = r_1(P,\B_p^d)^p = \max \{\|x\|_p^p : x \in P\}.$$ 
Thus tractability or hardness of Problems~\ref{prb:circumRadius} and \ref{prb:diameter} follow from Theorem~\ref{theo:mainTheo}.
\end{proof}

\paragraph{} Additionally, let $q\in [1, \infty]$ be such that $1/p + 1/q=1$ (with $1/\infty = 0$).

\begin{prb}[\textsc{Inradius}$_p$-$\CV$] \label{prb:inRadius}%
\begin{tabular}{ll}
\textbf{Input:} &$d \in \N$, $\gamma \in \Q$, rational $\CV$-presentation of a 0-symmetric polytope $P \subseteq \R^d$ \\
\textbf{Parameter:} & $d$ \\
\textbf{Question:}& Is $r(P,\B_q^d)^p \leq \gamma $?
\end{tabular}
\index{Inradius@\textsc{Inradius}$_p$-$\CV$}
\end{prb}

\begin{prb}[\textsc{Width}$_p$-$\CV$] \label{prb:width}%
\begin{tabular}{ll}
\textbf{Input:} &$d\in \N$, $\gamma \in \Q$, rational $\CV$-presentation of a 0-symmetric polytope $P \subseteq \R^d$ \\
\textbf{Parameter:} & $d$ \\
\textbf{Question:}& Is $R_1(P, \B_q^d)^p \leq \gamma $?
\end{tabular}
\index{width@\textsc{Width}$_p$-$\CV$}
\end{prb}

\paragraph{} As for the previous two problems, the question of $\NP$-hardness of \textsc{Inradius}$_p$-$\CV$ and \textsc{Width}$_p$-$\CV$ has been studied in \cite{gk-93}. It is shown that Problems~\ref{prb:inRadius} and \ref{prb:width} are solvable in polynomial time if $p=1$ and by using an identity for symmetric polytopes from \cite{gk-92} that both problems are $\NP$-hard when $p \in \N$. Here again, we can use the same identity to establish (in-)tractability for both problems when parameterized by the dimension:

\begin{cor}[Inradius {\&} Width]
For $p=1$, Problems~\ref{prb:inRadius} and \ref{prb:width} are in FPT. For $p\in \N\setminus \{1\}$, both problems are W[1]-hard.
\label{cor:widthHard}
\end{cor}
\begin{proof}
It is shown in \cite{gk-92} that if $P\subseteq \R^d$ is a 0-symmetric polytope and $P^\circ$ is its polar the identities
$$ R_j(P,\B_q^d)r_j(P^\circ, \B_p^d)= 1$$ hold for all $j \in [d]$.
As an $\CH$-presentation of $P^\circ$ is readily translated into a $\CV$-presentation of $P$, tractability or hardness of Problems~\ref{prb:inRadius} and \ref{prb:width} follow from Corollary~\ref{cor:radiiHard}.
\end{proof}

\paragraph{} The reductions of Corollaries~\ref{cor:radiiHard} and \ref{cor:widthHard} also show that the algorithm in the proof of Lemma~\ref{lem:approxFPT} can be used to compute the respective radii of a symmetric polytope $P \subseteq \R^d$ in the respective presentation. Lemma~\ref{lem:NoBetterApprox}, in turn, shows that in these cases the given running time is also best possible.

\bibliographystyle{plain}
\bibliography{references}
\end{document}

%% file: packageincludesE.tex
\usepackage[german, english]{babel}
\usepackage[pdftex]{color, graphicx}
\usepackage{dsfont}
\usepackage{amsmath}
\usepackage{amssymb}
\usepackage{bm}
\usepackage{url}
\usepackage{dsfont}
\usepackage[utf8]{inputenc}
\usepackage{enumerate}
\usepackage{ifthen}
\usepackage{paralist}
\usepackage{fancyhdr}
\usepackage{subfigure}
\usepackage{pdfpages}
\usepackage[babel, english=american]{csquotes}
\usepackage{makeidx}
\usepackage{geometry}
\usepackage{nccmath}
\usepackage{booktabs, multirow}
\definecolor{myblue}{rgb}{0,0,0.5}
\usepackage[a4paper, plainpages=false, colorlinks, 
    urlcolor=black,
    filecolor=black,
    citecolor=black,
    linkcolor=black,
    bookmarks,
    bookmarksnumbered,
    pdfpagemode=UseOutlines,
    pdfcreator = {LaTeX with hyperref package},
    pdfproducer = {pdfLaTeX}]{hyperref}

\newtheorem{theorem}{Theorem}[section]
\newtheorem{corollary}[theorem]{Corollary}
\newtheorem{proposition}[theorem]{Proposition}
\newtheorem{lemma}[theorem]{Lemma}
\newtheorem{remark}[theorem]{Remark}
\newtheorem{observation}[theorem]{Observation}
\newtheorem{definition}[theorem]{Definition}
\newtheorem{claim}[theorem]{Claim}
\newtheorem{problem}[theorem]{Problem}
\newtheorem{question}[theorem]{Question}
\newtheorem{notation}[theorem]{Notation}
\newtheorem{beispiel}[theorem]{Example}
\newtheorem{algorithm}[theorem]{Algorithm}
\newtheorem{conjecture}[theorem]{Conjecture}

\newenvironment{theo}[1][\empty]{\begin{theorem} 
\ifthenelse{\equal{#1}{\empty}} {}  {\normalfont(\itshape #1\normalfont)} \normalfont ~\\}{\end{theorem}}
\newenvironment{cor}[1][\empty]{\begin{corollary} 
\ifthenelse{\equal{#1}{\empty}} {}  {\normalfont(\itshape #1\normalfont)}  \normalfont~\\}{\end{corollary}}
\newenvironment{prop}[1][\empty]{\begin{proposition} 
\ifthenelse{\equal{#1}{\empty}} {}  {\normalfont(\itshape #1\normalfont)} \normalfont ~\\}{\end{proposition}}
\newenvironment{lem}[1][\empty]{\begin{lemma} 
\ifthenelse{\equal{#1}{\empty}} {}  {\normalfont(\itshape #1\normalfont)} \normalfont ~\\}{\end{lemma}}

\newenvironment{prb}[1][\empty]{\begin{problem} 
\ifthenelse{\equal{#1}{\empty}} {}  {\normalfont(\itshape #1\normalfont)} \normalfont ~\\}{\end{problem}}

\newenvironment{nota}[1][\empty]{\begin{notation} 
\ifthenelse{\equal{#1}{\empty}} {} {\normalfont(\itshape #1\normalfont)} \normalfont ~\\}{\end{notation}}

\newenvironment{proof}[1][\empty]{\ifthenelse{\equal{#1}{\empty}} {\paragraph{Proof.}~\\}  
{\paragraph{Proof of #1 }~\\}} {$\hfill \Box$}

\newcommand {\ones}[0] {\mathds{1}}

\newcommand {\N}[0] {\mathbb{N}}
\newcommand {\R}[0] {\mathbb{R}}
\newcommand {\Q}[0] {\mathbb{Q}}
\newcommand {\Z}[0] {\mathbb{Z}}
\renewcommand {\S}[0] {\mathbb{S}}

\newcommand {\B}[0] {\mathbb{B}}

\newcommand {\NP}[0] {\mathbb{NP}}
\renewcommand{\P}[0] {\mathbb{P}}

\newcommand{\conv}[0] {\mathrm{conv}}
\newcommand{\aff}[0] {\mathrm{aff}}
\newcommand{\lin}[0] {\mathrm{lin}}

\newcommand{\pos}[0] {\mathrm{pos}}
\newcommand{\ext}[0] {\mathrm{ext}}

\newcommand{\argmax}[0] {\mathrm{arg}\max}

\newcommand{\bd}[0]{\mathrm{bd}}

\renewcommand{\int}[0]{\mathrm{int}}
\newcommand{\relint}[0]{\mathrm{relint}}
\newcommand{\CA}[0]{\mathcal{A}}

\newcommand{\CH}[0]{\mathcal{H}}

\newcommand{\CP}[0]{\mathcal{P}}

\newcommand{\CV}[0]{\mathcal{V}}

\newcommand{\normmax}[1][\empty]{\ifthenelse{\equal{#1}{\empty}} {\textsc{Normmax}}  
{\textsc{Normmax}$_{#1}$}}

\renewcommand{\vector}[1]{\left(\begin{array}{c} #1 \end{array}\right)}